\documentclass{article}
\usepackage[utf8]{inputenc}
\usepackage{amssymb}
\usepackage{amsmath} 
\usepackage{amsthm}
\usepackage{tikz-cd}
\usepackage{authblk}
\usepackage{cancel}
\usepackage{url}

\newtheorem{theorem}{Proposition}

\title{Establishing a  pre-logical setting \\ in a quantum model of psychoanalytic theory\\
}
\author{Giulia Battilotti}
\author{Rosapia Lauro Grotto}
\affil{Dept. of Health Sciences, University of Florence \\
giulia.battilotti@icloud.com - rosapia.laurogrotto@unifi.it}

\date{\today}

\begin{document}

\maketitle

\begin{abstract}
A crucial issue both in cognitive and psychoanalytical theories deals with the origin of mental representations.
In order to explore this issue, the paper analyzes a pre-logical setting, by considering a formalized approach to the foundations of psychoanalysis in logic,   interpreting and integrating the views  by Freud, Matte Blanco, Klein and Bion. The formalized approach derives from a quantum model of spin states.
    A representation of the spin state of a particle in first order logic is abstracted to get a modality interpretable as an abstract projector. The last can be decomposed into a positive, negative and irreal component. The irreal component cannot emerge and, in logic, is absorbed by the two others, giving rise to logical duality. Due to its treatment of undefiniteness and coherence, the paper is meant to contribute to quantum cognition, in its particular sense of affective quantum cognition.
\end{abstract}


\section{Introduction}

The recent development of formal approaches to psychoanalytic theory seems to provide a way to develop a new model of mental processes. This is based on the central idea that it is possible to describe in formal terms the dynamics of representations endowed with affects, which play a central role both in psychoanalytic theory and in cognitive approaches.
In fact, psychoanalysis and contemporary cognitive sciences, the latter including both the classic ‘cold’ approach and the contemporary paradigm of embodied and affective cognition, do indeed share a common focus regarding representation and emotion \cite{Lg20}. The importance of affects in cognitive processes was first stressed long time ago by Vjgotskji, who said that there is an affect behind every thought.

 Therefore the aim of our multi-disciplinary research is to investigate about a logical basis 
to develop formal approaches in cognitive studies and their applications, by building on the foundations of psychoanalytic theory, that is the Freudian theory of psychic representations. 
Then the paper advances  a formalized approach to the theoretical foundations of psychoanalysis in logic. In particular, our proposal can support the links of psychoanalysis and Artificial Intelligence, that are gaining an increasing interest, see \cite{Po}.

The paper characterizes a pre-logical setting, interpreting and integrating different theoretical  views developed in psychoanalysis. 
Psychoanalysis deals with particular abstract issues of our mind very important in order to understand human thinking: the issues of  coherence vs separation, distinguishability vs indistinguishability, definiteness vs undefiniteness. Indeed,  the issues  are well present in the construction of our thinking, since they are rooted in its unconscious origin. 
Our analysis is founded on such abstract issues of thinking, and aims at clarifying them, in logical terms. 

The natural environment in which our proposal is developed  is a quantum model, since the quantum world offers a unique model to treat such kind of issues, as  widespreadly witnessed by the literature. 
Actually, in our model, the point is to  grasp the concept of quantum state itself. 
In the  paper we adopt a logical representation of quantum states in first order language \cite{Ba14}, and its successive abstraction by means of modalities \cite{BBL2, BBL3}. The method comes from basic logic \cite{SBF}, and consists of defining the logical constants by putting suitable equations which can import the metalinguistic links into sequents. The equations we put allows for an analysis of the modalities that follows from the algebraic splitting of the spin observable into components given by the Pauli matrices. Overall, the infinite/indefinite vs finite features of the modalities are considered and discussed. 

In order to establish the correspondence between the quantum model and the psychoanalytic model, 
we  first see the correspondence of the representation of quantum states with Freud's very idea of representation \cite{Fr91}, namely Freud's proposal of how the mind  creates its  own objects. The existence of such a correspondence allows for the extension of the quantum model to successive developments of Freudian theory,  proposed  by Freud himself \cite{Fr00,Fr23,Fr25}, Klein \cite{Kl}, Bion \cite{BiL, BiE}, Matte Blanco \cite{MB75,MB88}, in which the logical/structural aspects of the  mind are in evidence. In particular, we adopt as a basis Matte Blanco's finding that the structural feature of our thinking as originated by the Unconscious relies in the infinite \cite{MB75}. His proposal, for us, has been an assist to our goal:  to develop a formalized view of the Structural Unconscious, as characterized in The Interpretation of Dreams \cite{Fr00}. Moreover, the pre-logic we aim to introduce avails of the successive development of the Freudian Second Topic \cite{Fr23}, in order to introduce the normative element on which logic itself is founded. The discussion of modalities finds a particular correspondence with the seminal paper {\em Negation}, published exactly 100 years ago by Freud \cite{Fr25}.  Then, we consider the theory of Object Relations, as introduced by Klein \cite{Kl} and developed by Bion in the so-called epistemic trilogy \cite{BiL,BiE,BiT}, to analyze how our mind develops psychic representations and hence structures its knowledge. All the formalism is developed consistently with the quantum model.

The present proposal continues the previous papers \cite{BBL1,BBL2,BBL3,BBL4}. It can be located in the framework of Affective Quantum Cognition. The field of Quantum Cognition \cite{PB} has  increasing applications. Its particular aspect of Affective Quantum Cognition, whose theoretical formal basis can be traced back to the pioneering proposal  \cite{Kh02}, see then \cite{Iu18} and \cite{IKM}, and that finds support on Quantum Theories of Consciousness \cite{HP}, is now applied in A.I.,  see the review paper \cite{YU}. In particular, one could consider the recent application \cite{HH}, that {\em provides a system for simulating and handling affective interactions among various agencies from an understanding of the relations between quantum algorithms and the fundamental nature of psychology.}
Indeed, 
  the Structural Unconscious, as conceived by Freudian Theory,  is like an agent, that is able to carry positive effectiveness into  decision processes. This fact has been independently experimentally tested in psychology time ago, \cite{Di} and \cite{AD}, and it is interesting to ask oneself why and how. By analyzing a pre-logical framework able to mediate between the logic of the Unconscious and rational thinking, our work would like to support,  in particular, the introduction of this kind of new agent in formal applications to decision making.

 \section{Coherence conditions}

In building our pre-logical setting, we aim to meet the view of logic that F.Enriques proposed more than a century ago, in \cite{En06}. In his words\footnote{“Riconosciamo, ad ogni modo, che la Logica può riguardarsi come un insieme di norme, le quali {\em debbono}
osservarsi, {\em se si vuole} la coerenza del pensiero. Ma
ciò può anche essere espresso dicendo, che: fra i varii
procedimenti mentali, se ne distinguono alcuni, in cui
vengono volontariamente soddisfatte certe condizioni
di coerenza, i quali si denominano appunto procedimenti logici.
{\em In questo senso la Logica può riguardarsi come una
parte della Psicologia.}" }:

    “Anyway, we recognize that Logic can be regarded as a set of norms, which {\em must} be observed {\em to the aim of} coherence of thinking. On the other hand, this can also be expressed by saying that: among the different mental procedures, some are characterized, in which certain coherence conditions, that indeed are termed logical procedures, are willingly satisfied. {\em In this sense, Logic can be regarded as a part of Psychology.}". Italics in the text.

Coherence of thinking cannot be referred only to the logical-deductive processes, but it has to take into consideration the way in which mental objects are built. In this perspective,  among the different theories in psychology, we have started from the psychoanalytic conceptualization centered on the very concept of psychic representation.
So let us start from the idea that coherence is required in order to characterize mental objects dealt with by logic.
According to Freud, mental objects are first characterized by 
{\em thing-presentations}, 
a concept at the basis of his theorization, \cite{Fr91}. 
Thing presentations are non-verbal open representations of objects operated by the Unconscious.
Thing presentations can access consciousness only when closed by words. 
Then,  {\em word presentations} are the closed representation of objects managed by consciousness. They require  pre-existing thing-presentations.
\subsection{Coherence by infinite singletons}
How to represent, formally, a thing-presentation? We investigate about coherence  on the basis of the idea of variable on a domain.
In order to include thing-presentations in logic, we consider quantified formulae on non-extensional domains termed {\em infinite singletons} \cite{Ba14b,BBL1}.  
Infinite singletons are characterized intensionally rather than extensionally:
one says that $V$ is a singleton if and only if it satisfies the equivalence 
\begin{equation}\label{infsingl}
(\forall x\in V)A(x)\equiv (\exists x\in V)A(x)    
\end{equation}
for every formula $A$,
whereas the consequence:
\begin{equation}\label{extsingl}
  z\in D\vdash z=u  
\end{equation}
for some closed term $u$, is not assumed. 

The idea of infinite singleton can be derived in a direct way from  Matte Blanco's logical characterization of the Structural Unconscious, proposed in his {\em The Unconscious as infinite sets} \cite{MB75}. He characterized the mode of the Unconscious (Symmetric Mode) by two principles: 
\begin{itemize}
\item the Symmetry Principle 
\item the Generalization Principle
\end{itemize}
As discussed in \cite{MB75}, among the consequences of the two principles, one finds that 
all relations are symmetric
and all sets are infinite
for the Unconscious.
If two different elements in a set are characterized, an order can be put.
The necessary conclusion is that the Unconscious operates on infinite singletons. 
This means:
The Unconscious cannot characterize! Such a feature is coherent with the features of the mental representations as we observe them in dreams, as well as with the typical symptoms of schizofrenic thinking, \cite{MB75}, \cite{BBL1}.

Freud characterized the process of the Unconscious (Primary Process) in The Interpretation of Dreams \cite{Fr00}. 
The Primary Process has the following features:

\begin{itemize}

\item Displacement;
\item Condensation;
\item Absence of contradiction;
\item Substitution of the external reality with the internal one;
\item Timelessness.
\end{itemize}
\noindent

Let us consider the first two features, displacement and condensation, that  are the mental procedures possible only in presence of the strong coherence allowed by the Unconscious. As discussed in \cite{BBL1}, we see that they have an immediate translation in terms of  infinite singletons.
Displacement means  the displacement of a property from an element to “another”. 
It occurs in any infinite singleton. For, if $z\in U$ and $A(z)$, then $(\exists x\in U)A(x)$, but then $(\forall x\in U)A(x)$, by definition of infinite singleton.
Condensation means that any two “different" objects can condense into a unique one. 
Any two infinite singletons $U$ and $V$ cannot be distinguished in the symmetric mode. For, $z\in U$ and $z\notin V$ is impossible, since the Unconscious can establish a membership, but it cannot exclude one, by the generalization principle.  Therefore, any two infinite singletons are forced to condense together. 
Then,  when  the Symmetric Mode finds no obstacle at all, one gets  a unique object and thinking is impossible, see \cite{MB88}. So infinite singletons create a too strong coherence.
\medbreak
The first question is then: 
\begin{center}
{\em How  to avoid the spreading of coherence?}    
\end{center}
\noindent According to Freud, \cite{Fr91}, 
the closure of the representations is produced by the characterization of an external reality. As discussed in \cite{Lg20}, the mediation with the Symmetric Mode, namely the relation of mental processes with the external reality, can happen with the support of  the so-called “transitional objects", whose abstract features are that they must be found in the real external world, on one side, however they must have an a-temporal, a-causal nature, on the other, see \cite{Wi}, quoted in \cite{Lg20}.
Once the external reality has been established for the mind,  objects are defined and separated. So coherence stops since representations are closed. 
When representations are closed, infinite sets {\em unfold}  (Matte Blanco's terminology, see \cite{MB75}),  into well-defined descriptions of  objects, leading to the conception of  finite sets and singletons, see \cite{Lg20} for a discussion.

\subsection{Looking for the right amount of coherence in logic}
Let us  have a closer look at the translation of the above points  in first order language, in which infinite singletons are characterized.
In the logical language, constants are adopted for word-presentations,  when objects are defined, in the conscious process of thinking. On the contrary, let us assume that variables are the witness of an abstract pre-existing attitude, that is even present prior  to representations themselves, as we will see later, and that can be also consciously recovered after the representations. Then, in the logical language, let us consider quantifiers, adopting the definitions
introduced in  \cite{MS}, obtained by putting suitable equations,  as in basic logic \cite{SBF}.

The definition of universal quantifier on a domain $D$ is the following: 
\begin{equation}\label{defforall}
  \Gamma (-z)\vdash (\forall x\in D) A(x) \quad\equiv\quad \Gamma (-z), z\in D \vdash A(z)  
\end{equation}
where we adopt the notation $\Gamma (-z)$ to mean that $\Gamma$ is closed with respect to the free variable $z$, that is a variable of the language.
In the following (subsection \ref{duality}), we shall need also the dual definition of existential quantifier, given as follows (see \cite{MS}, and \cite{SBF} for duality):
\begin{equation}\label{defexists}
 (\exists x\in D) A(x) \vdash \Delta(-z)  \quad\equiv\quad A(z), z\in D \vdash \Delta(-z)  
\end{equation}

The definition of universal quantifier is compared to the definition of the “omega-quantifier" $\forall_\omega$:
\begin{equation}\label{defforallomega}
  \Gamma \vdash (\forall_\omega x\in D) A(x) \quad\equiv\quad \Gamma \vdash A(t)\;\mbox{{\em forall}} \;t\in D  
\end{equation}
where $t$ is a closed term of the language denoting a parameter (for the sake of simplicity we adopt the same notation for the element of $D$ and the closed term denoting it in the object language). 
Notice that defining $\forall_\omega$ on the domain $D$ requires that $D$ is described by a set of closed terms of the language. In particular, if $D$ is described by a finite set of $n$ terms $t_1,\dots, t_n$, the equation defines the propositional  conjunction of $n$ formulae $A(t_1)\&\dots \& A(t_n)$, that is the additive conjunction in terms of linear logic. If $D$ is described as a singleton, by a closed term $u$, $(\forall_\omega x\in D) A(x) $
is $A(u)$. The dual of $\forall_\omega$ is $\exists_\omega$, defined as follows:
\begin{equation}\label{defexistsomega}
(\exists_\omega x\in D) A(x) \vdash \Delta  \quad\equiv\quad A(t)\vdash \Delta  \;\mbox{{\em forall}} \;t\in D      
\end{equation}
In particular,  $(\exists_\omega x\in D) A(x)$ is the additive disjunction $A(t_1)\vee\dots \vee A(t_n)$ if $D$ is described by a finite set of $n$ terms $t_1,\dots, t_n$, it is $A(u)$ If $D$ is described as a singleton.
One can prove that:

\begin{theorem}\label{propequiv}
$(\forall x\in D) A(x)$ and $(\forall_\omega x\in D) A(x)$ are equivalent formulae if and only if the membership relation $z\in D$ is fully described by closed terms, that is one has $z\in D\vdash (\exists_\omega x\in D)x=z$. 
\end{theorem}

\begin{proof}
We first observe that, since the sequent $(\forall x\in D) A(x), z\in D\vdash A(z)$ is always derivable, by substitution (it follows from equation \ref{defforall} putting $\Gamma=(\forall x\in D) A(x)$,  substituting $z/t$ in it and then cutting the true premise $t\in D$),  one has $(\forall x\in D) A(x)\vdash A(t)$ {\em forall} $t\in D$, and hence the sequent
   $(\forall x\in D) A(x)\vdash (\forall_\omega x\in D) A(x)$ is always derivable, by the definition \ref{defforallomega} of  of $\forall_\omega$.
   Similarly, $(\exists_\omega x\in D)x=z \vdash (\exists x\in D)x=z$ is derivable from  $z\in D, x=z \vdash (\exists x\in D)x=z$ by substitution $z/t$ and definition of $\exists$.

   Let us assume $z\in D\vdash (\exists_\omega x\in D)x=z$, and let us consider the set of sequents $(\forall_\omega x\in D) A(x)\vdash A(t)\;\mbox{{\em forall}} \;t\in D $, that are derivable by the definition \ref{defforallomega}. Then $(\forall_\omega x\in D) A(x), t=z\vdash A(z)\;\mbox{{\em forall}} \;t\in D $, by the equality rules. 
   Then $(\forall_\omega x\in D) A(x), (\exists_\omega x\in D)x=z\vdash A(z)$ by definition of $\exists_\omega$ above. Then, by the assumption $z\in D\vdash (\exists_\omega x\in D)x=z$, one has  $(\forall_\omega x\in D) A(x), z\in D\vdash A(z)$, from which $(\forall_\omega x\in D) A(x)\vdash (\forall x\in D) A(x)$ by definition \ref{defforall} of $\forall$.

   Let us assume $(\forall_\omega x\in D) A(x)\vdash (\forall x\in D) A(x)$, that means $(\forall_\omega x\in D) A(x), z\in D \vdash  A(z)$ by definition of $\forall$ and put $A(x)\equiv x\neq y$ in the last. By duality one gets $z=y, z\in D\vdash (\exists_\omega x\in D)(x=y)$, then, by definition of $\exists$,
   $(\exists x\in D)x=y\vdash (\exists_\omega x\in D)x=y$, that means $y\in D\vdash (\exists_\omega x\in D)x=y$, since one can see that $(\exists x\in D)x=y$ and $y\in D$ are equivalent.
\end{proof}
Then we maintain that the definition of $\forall$, where the variable is internal, can form an infinite logical object, whereas the definition of $\forall_\omega$, where the variable is a parameter, can form a finite logical object: this is independent of the nature, finite or infinite, of the quantification domain $D$, as considered at the metalevel. The infinite rather than finite nature of the domain $D$  is given by the adoption of the variable in the object language. 

All the above facts  reveal the important yet somewhat hidden consequence that our usual  view of singletons is too “finitistic". For, one usually assumes that a singleton  is extensionally characterized by its element and hence adopts a closed term $u$ for it, assuming 
\ref{extsingl},
and  hence identifying the representation $(\forall x\in D)A(x)$ with $(\forall_\omega x\in D)A(x)$, that is $A(u)$. However, such a description may be only partial with respect to the original thing one should grasp. A set of words (in an ideal setting possibly infinite, this is not the point here, as specified above) $\{u_1,\dots ,u_n, \dots\}$ might be more appropriate in order to find a complete description of such a thing. Assuming $x\in D\vdash (\exists_\omega x\in D)A(x)$,  the description so obtained would be $(\forall_\omega  x\in D)A(x)$, that, in presence of the assumption, would be equivalent to $(\forall x\in D)A(x)$, as one can prove (see \cite{BBL2}). However, the assumption might be not the case. So, the original description $(\forall x\in D)A(x)$ is not recovered.
\medbreak
So the second question is: 
\begin{center}
{\em How to establish a mental coherence, from the different items coming from the contact with reality? }    
\end{center}

\subsection{Partial and total objects in psychoanalytic theories}\label{psych}

In order to fully comprehend which kind of contribution the psychoanalytic conceptualization can give, in order to establish  a pre-logical setting, we need to recall that objects, as dealt with by psychoanalysis, are  representations always equipped with affects, whose origins can be traced back to the first interactions of the infant with their caregiving environment \cite{Lg20}.

A fundamental contribution on how a mental object can be created comes from the Kleinian theory:
according to Klein \cite{Kl}, the first contact with the external reality is characterized by the separation between the positive-valued experiences ({\em good object}) and the negative-valued experiences ({\em bad object}). Experiencing such an original characterization, the infant's mind emerges from the condition of an indifferentiated and fusional relationship with the external reality. As a consequence the construction of a mental object must necessarily include the integration of the previously splitted and hence separated components ({\em good object} and {\em bad object}) into a {\em total object}.

So, thoughts are derived in the integration of previously non integrated elements: it is the switching from the so-called {\em Paranoid-Schizoid Position} to the {\em Depressive Position}. 
The Kleinian Theory, by according a preference to the term {\em Position} rather than adopting “developmental stage", stresses how the integration processes are constantly alternated to processes in which total representations are re-splitted, in order to foster the emergence of new aggregations. Then, two primitive abstract opposite attitudes should be implemented in our mind: the separating and the  integrating one. 

The alternation of the two movements, integration and splitting, usually denoted by $PS \leftrightarrow D$, has been recognized as the basic dinamics of the processes of thoughts by Bion,  who's theory refers to the process of learning from  experience, that is, to the way in which new concepts ideas can be derived by the mind in contact with external reality, see \cite{BiL}. According to Bion, the contact with an unknown reality can induce the Preconception of the existence of an object; once we assume an object is present a mental Container is created where all experiences about that object (the Contained) can be allocated. A relevant feature of the theory concerns the hypothesis, derived from clinical observations, that the Container and the Content are linked in reciprocal interaction \cite{BiL}, and the fitting of the two cannot be given for granted. Then the theoretical description of  Container-Contained interaction includes all the possible results of the contact with reality, even if it is not confined to them. In the contact with reality, three different cases can be given:
\begin{itemize}
    \item The realization corresponds to the preconception: it is the positive case, a concept can be derived from the contact with reality. The positive case corresponds to the acceptance of the representation of reality one has achieved by means of infinite singletons. It is made possible by the {\em Convivial Link} in the 
    {\em Container-Contained} interaction.
    \item The realization does not correspond to the preconception: negative case, the representation must be rejected.
    The negative case is the rejection of  the representation of  reality one has achieved, that means its repression, in Freud. 
    \item The experience of reality is unbearable, therefore the possibility to create a representation is destroyed. The so called “attack to the link", or {\em Parasitic Link}, is created in the {\em Container-Contained} interaction, yielding 
    a failure of the contact with reality, entailing in turn a failure of the process of representation.
\end{itemize}

\section{Quantum model and modal projector}
Going back to our two questions about coherence, in the following, we introduce a proposal to answer to the first question which allows for an analysis of the second question too. 

The first question is addressed by Freud in the moving from his First Topic, \cite{Fr00},  to his Second Topic, \cite{Fr23}. According to Freud, the obstacle to the spreading of symmetry is supplied by the introduction of a normative instance moderating the encounter of the psychic dimension with the external reality.
In logic, normativity is described by means of modal operators.
In the following, we model the normative dimension of the theory by considering a quantum approach, since the quantum world is a natural model to discuss coherence and its confinement.
Our basis is the representation of quantum states in first order language \cite{Ba14}, and the successive introduction of the abstract projector as a modal abstraction in the spin model \cite{BBL2}.

\subsection{Infinite singletons and quantum spins}
As seen in \cite{BBL2}, the quantum measurement of the state of a quantum particle $\cal A$ with respect to a given spin observable $\sigma_d$ (direction $d$) is associated to an equation of the form
\ref{defforallomega}. The domain $D=D_d$ is the set of outcomes of the spin measurement (with the associated probabilities) along the direction $d$:
$D_d=\{t_1,\dots ,t_n\}$: for the spin measurement, $n=1$ or $n=2$ and $t_i=(s_i,P\{X_d=s_)\})$, where $s_i\in \{\uparrow_d, \downarrow_d\}$ 
and $P\{X_d=s_i)\}>0$). Hence the mixed state after measurement is characterized by the formula
$$(\forall_\omega t\in D_d)A(t)$$ 
By importing the parameter as an internal variable of the language, we have equation \ref{defforall}, defining the universal quantifier, with $D=D_d$. 
$D_d$ is an infinite singleton prior to measurement, according to proposition \ref{propequiv}, since the equivalence between $z\in D_d$ and $(\exists_\omega x\in D_d)x=z$ is not true, see \cite{BBL2}. The pure state is then described by the universal formula:
$$
(\forall x\in D_d)A(x)
$$
Moreover, the following sequent describes  the quantum measurement of the particle:
$$(\forall x\in D_d)A(x)\vdash (\forall_\omega t\in D_d)A(t)$$
It is derivable by substitution of the variable by the closed terms denoting the elements of $D_d$, see \cite{BBL2}.
The converse sequent is not true. This fact is an aspect of incompleteness, both in logic and in physics. In logic, it allows for the unprovability of G\"odel's diagonal sentence \cite{Go31}. In quantum physics, it corresponds to the irreversibility of the measurement, so that the original information contained in the pure state is not recovered: as is well known, a fact long debated in the foundations \cite{Ja}.

\subsection{The Bloch Sphere and its constants}\label{Bloch}
Let us consider the spin model pictured in the Bloch Sphere: as is well known, in the sphere, the  points  of the surface  correspond to projectors (pure states), in particular any two antipodal points in the direction $d$ represent the two eigenvectors of the spin observable $\sigma_d$, while the inner points correspond to the convex combinations of projectors, namely they correspond to mixtures.

\begin{center}
     \includegraphics[width=1.0\textwidth]{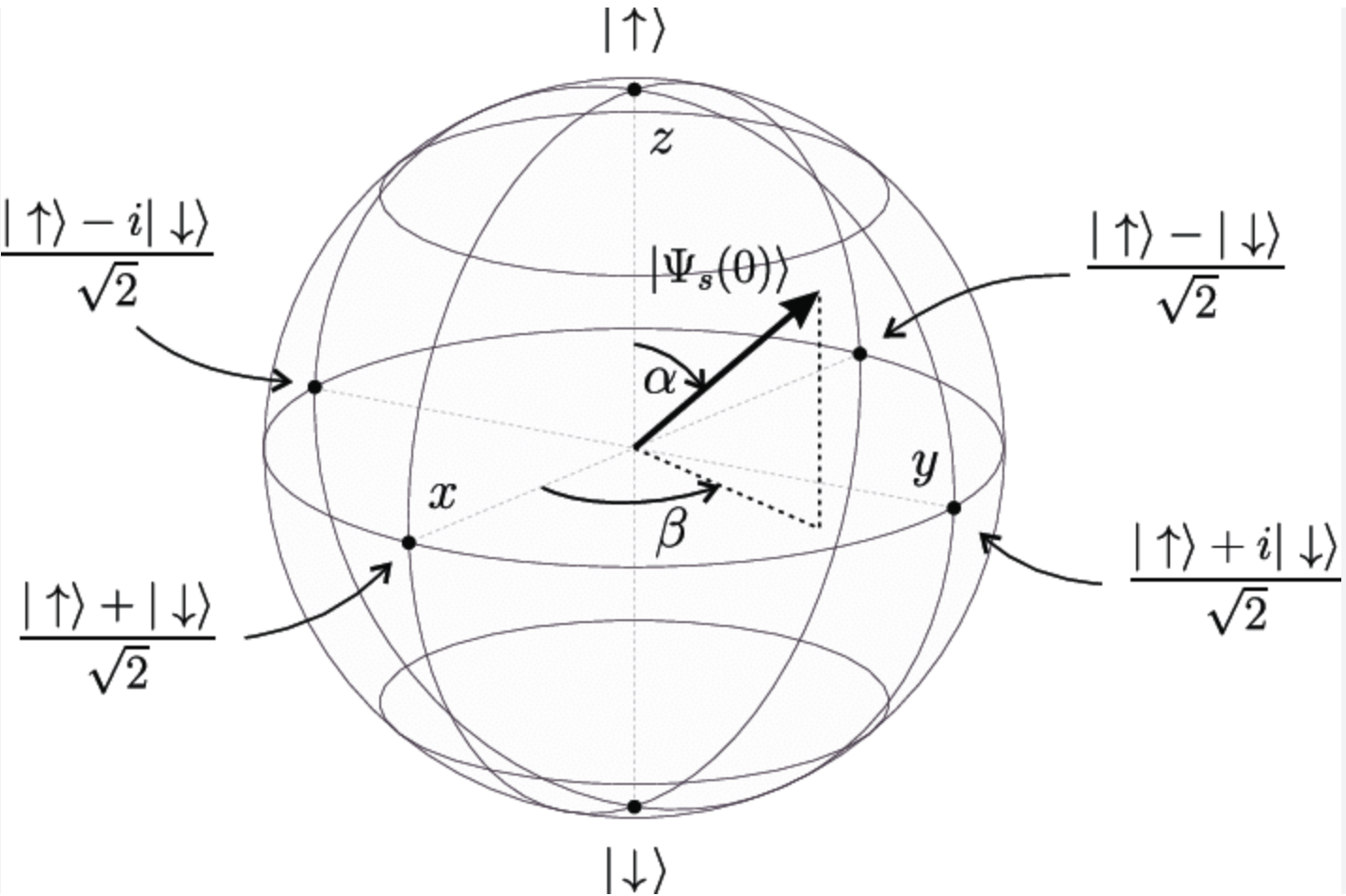}
     \end{center}

Ipso facto, considering the Bloch Sphere, two logical constants are created. The first corresponds to the maximally mixed state given by the couple of eigenstates of any observable $\sigma_d$, described by the convex combination $\dfrac{1}{2}P_\downarrow+\dfrac{1}{2}P_\uparrow$, that is the same for every $d$. So we characterize such a state by a constant, $\perp$. It is associated to the center of the Bloch sphere, no information.  It represents the totally-non integrated, separated element.
 
The second corresponds to the singlet state $\dfrac{1}{\sqrt{2}}|\uparrow \downarrow>-\dfrac{1}{\sqrt{2}}|\downarrow\uparrow>$, that  is independent of $d$ as well.  One might observe that the singlet state is attributed to a couple of particles and hence it does not fit in the Bloch Sphere. However, the two particles of the couple are indistinguishable, totally dependent  on each other (see \cite{Kr}) and together form a pure state independent of $d$. So let us associate the singlet state with the whole surface, the whole information, denoted by the constant $1$.   It represents  the totally integrated, non-separable object.

Let us assume that $\perp$ and $1$ represent two abstract entities for the separating and the integrating position, respectively, namely the pure Paranoid-Schizoid and the pure Depressive Position proposed  by Klein, see subsection \ref{psych}.  Let us  consider an operator $*$ describing the switching from one Position to the other: $\perp^*=1$ and $1^*=\perp$. In the following, we shall see that $*$ corresponds to a duality.

\subsection{Introducing the modal projector}
\iffalse
Assuming Freudian theory, any mental function, including judgement, is first operated by the Unconscious. Assuming that our first judgements derive from creating infinite singletons and adopting them to form predicates  derived from assuming equations  like \ref{defforall}, we should say that no privileged direction is characterized. So, let us consider the whole Bloch sphere, not only one direction, to grasp the function of judgement.
\fi
Let us assume that we can characterize a particle $\cal A$ in a quantum system. Then its spin measurement in direction $d$ finds a domain $D_d$, the description of the associated  mixed state as $(\forall_\omega x\in D_d)A(x)$ and the description of the pure state as $(\forall x\in D_d)A(x)$, by
 the above definitions \ref{defforallomega} and \ref{defforall}.
 In the Bloch Sphere a dot on the surface and an inner dot are characterized. Let us abstract with respect to $d$, so that no privileged direction is adopted. Dropping the domains $D_d$ and the closed terms and variables for its elements from \ref{defforall}, we get the unique form
\begin{equation}\label{defmod}
\square \Gamma\vdash \square A\;\equiv\; \square \Gamma\vdash A    \end{equation}
Observe that the same definition is an abstract form for  \ref{defforall} and for \ref{defforallomega} as well. Then, 
$\square$ is like an abstract quantifier, located in between the infinite quantifier $\forall$ and the finite one $\forall_\omega$.

One can prove that equation \ref{defmod} defines the modal operator $\square$ of S4 (the necessity operator, see \cite{BBL2}). Then $\square$ satisfies
\begin{equation}\label{S4}
 \square\square A=\square A   
\end{equation}

Then, in quantum terms, $\square$  can be interpreted as an abstract projector, the {\em modal projector}.  Since it is located in between $\forall$ and $\forall_\omega$, it can have both an infinite/undefined and a finite/defined interpretation: namely, underneath, it can depend on an internal variable or it can gather  externally parametrized objects.
So, it can attribute a  sharp yet undefined state to the particle.

\section{Components of the modality and their interpretations}

The above equation \ref{defmod} is the result of an abstraction on the variable $d$ of \ref{defforall} applied when the observable $\sigma_d$ is considered, namely: 
$$
\Gamma(-z)\vdash (\forall x\in D_d)A(x) \;\mbox{if and only if}\; \Gamma(-z), z\in D_d \vdash A(z)
$$
Then, $\square$ itself is considered as an abstract quantifier, applied to its own domain $T$, that is the container of all the results $s(d)$ given by the variation of $d$. We conceive $T$ as an infinite singleton, the total infinite singleton, generalizing the particular ones created by each  direction $d$. Then we should have
\begin{equation}
\square A\equiv (\forall x\in T)A_\forall (x)   
\end{equation} where $A_\forall (d)$ abbreviates $(\forall x\in D_d )A(x)$: in the following we write simply $A(d)$. For, if conceived as a particular quantifier, 
$\square A$ is given by the equation, analogous to \ref{defforall}:
\begin{equation}\label{defmodinf}
\square\Gamma\vdash \square A \;\mbox{if and only if}\; 
\square\Gamma, s(d) \in T\vdash A(d)
\end{equation}
where $\square$ includes the closure with respect to the variable $d$ and all the variables for the domains $D_d$. In such an interpretation, $\square A$ is the unique, infinite and complete object capturing the unique, sharp yet undefined state of the particle. 

On the other hand, equation \ref{defmod} extends equation \ref{defforallomega} too, and we have the following equation producing the finite interpretation of $\square A$: 
\begin{equation}\label{defmodfinh}
\square\Gamma\vdash \square A \;\mbox{if and only if}\; 
\square\Gamma\vdash A(t) \;\mbox{forall}\; t: s(d(t))\in T
\end{equation}
Here we can conceive $t$ as a temporal parameter to which the direction $d=d(t)$ is associated (we can write $\square \Gamma$ in the premises since the hypothesis, that is the preparation of the quantum particle, is always the same and hence independent of $d$, and independent of all the variables for the domains $D_d$, then it is closed with respect to all the variables, see \cite{BBL2}).

This implies that the domain $T$ has also a finite view and so one has $\square A=(\forall_\omega x\in T)A(x)$. Actually, it has a lot of finite views, since the externalization of the temporal parameter creates an order and hence the object so formed can be considered to be sensitive to the initial state of the particle. This is better seen by considering that the spin observable is algebrically  represented as a Hermitian matrix in the space of $2\times 2$ complex matrices, and that, in linear algebra, the process of application of a variable Hermitian matrix with respect to a constant vector, considered in \ref{defmodfinh}, is equivalent to the process of application of a fixed initial Hermitian matrix with respect to a variable vector: as is well known the last is termed the “Schroedinger Picture" of the process. Then  equation \ref{defmodfinh} can be rewritten in Schroedinger picture, putting in evidence the fixed observable $\sigma_d$ as an assumption,   and the dependence on $d$ of the operator $\square_d$ so created:
\begin{equation}\label{defmodfin}
\square\Gamma\vdash \square_d A \;\mbox{if and only if}\; 
\square\Gamma, \sigma_d \vdash A(t) \;\mbox{forall}\; t: s(t)\in T  
\end{equation}
where $s(t)$ describes the variation of the state of the particle with time, for example on the surface of the Bloch sphere, so that $s(t)$ can be associated to a direction and hence the set $T$ can be adopted again as a domain. Here $A(t)$ is to be interpreted as $(\forall x\in D_t)A(x)$, where $D_t$ is the domain characterized by the observable $\sigma_d$ applied to the particle, that is in state $s(t)$ at time $t$. Notice that $\sigma_d$ is a strong assumption, that can change the quality of the operator $\square_d$ associated to it, as we shall see. 
Moreover, it is important to notice the following:
\begin{theorem}
 $\square_d A$ is a closed formula, namely    
 $\square_d A=\square \square_d A$ for every $d$.      
\end{theorem}
\begin{proof}
In \ref{defmodfin}, the closure with respect to all the elements of $T$ that allows to derive $\square_d A$ has been applied.     
\end{proof}

Now the nature of $T$ and $\square$, finite and infinite, and the different finite ways in which $T$ and $\square$ unfold, can be analyzed considering the algebraic decomposition  of the generic spin observable, represented as an element of the subspace of the Hermitian matrices in the space of $2\times 2$ complex matrices: any Hermitian matrix $\hat O$, can be decomposed as the real linear combination of four components: the three Pauli matrices $\sigma_X, \sigma_Y, \sigma_Z$ corresponding to a set of three orthogonal spin directions (say $x,y,z$ respectively),  and the identity $I$. Namely:
\begin{equation}\label{decomp}
\hat {O}= \alpha I+ \beta_x \sigma_X+ \beta_y \sigma_Y + \beta_z \sigma_Z    
\end{equation}
\iffalse
that is 
$$
 \hat{O} = \begin{pmatrix}
 r_1 & \beta_x-i\beta_y \\
 \beta_x+i\beta_y & r_2
 \end{pmatrix}
 =
 \alpha 
 \begin{pmatrix}
 1 & 0 \\
 0 & 1
 \end{pmatrix}
 +
 \beta_x
 \begin{pmatrix}
 0 & 1 \\
 1 & 0
 \end{pmatrix}
 +
 \beta_y
 \begin{pmatrix}
 0 & -i \\
 i & 0
 \end{pmatrix}
 +
 \beta_z
 \begin{pmatrix}
 1 & 0 \\
 0 & -1
 \end{pmatrix}
 $$
 \fi
where $\alpha, \beta_x, \beta_y, \beta_z$ are real numbers (see \cite{AFP}).

 The decomposition provides us four different basic  ways to give a content to
$T$ and an identity to the operators $\square_d$ defined by \ref{defmodfin}, as we see. 
\subsection{Again the infinite case}
The case $\alpha =1$ and $\beta_x=\beta_y=\beta_z=0$ in \ref{decomp}, that is the identity case,  is not associated to a fixed direction $d$, since the eigenvectors of the identity are in all directions. 
This means that  the identity adopted as an observable could enable us to know all the truth concerning the state of the particle, in all directions, at the same time. Namely, we should put “all directions" instead of one direction $d$ in the assumption $\sigma_d$ of the equation \ref{defmodfin}. That  is, the assumption becomes like the internal assumption $d\in T$ of equation \ref{defmodinf}: the external parameters are internalized as a variable. Then we find 
the infinite interpretation \ref{defmodinf} of the equation \ref{defmod} defining $\square$ and the infinite interpretation of $T$ as an infinite singleton, supplying  a complete answer in order to characterize the state of the particle. However, no finite description can be given. Moreover notice that, by decomposition \ref{decomp}, if $\alpha\neq 0$, the generic operator $\square_d$ defined by \ref{defmodfin} has an infinite component. 
In Bionian terms, such a feature corresponds to the {\em Symbiotic Link} of the {\em Container-Contained} relationship.

\subsection{The positive and the negative finite cases}\label{zx}
The three remaining elements $\sigma_X,\sigma_Y,\sigma_Z$ of the basis of the space of $2\times 2$ complex matrices are associated to the three orthogonal directions $z,x,y$.  
We first discuss the case $\beta_z=1$, that is  the Hermitian matrix $\hat O$ is $\sigma_Z$, that is $d$  coincides with the direction, say $z$, chosen for the preparation of the particle itself. This means that any measurement  gives back the “exact" probabilities for the state of the particle; in particular, if the particle has been prepared in the basis state “up", the answer is always “up", if the  particle has been prepared in the basis state “down", the answer is always “down". The  matrix $\sigma_Z$ can be written as the linear combination of the two projectors on the basis vectors, the projector on the vector “up", and the projector on the vector “down". Then, overall, the modality $\square_z$ defined putting $\sigma_Z$ in \ref{defmodfin}, creates a unique abstract projector, namely one has $\square_z\square_z A=\square_z A$, that is the definition of projector.
Then $\square_z$ satisfies the clause characterizing $S4$.
For this reason, we associate such a case with the finite interpretation of the equation \ref{defmod} and write simply $\square$, since equation \ref{defmod} is unique. We recall that, in quantum mechanics, any projector is associated to a fixed vector.  Here, let us associate the abstract projector   to an  abstract element, say $p$. $p$ is the abstract positive witness, namely the element that enables to evaluate the state of the particle consistently with the data from reality. Then we interpret $T$ as the finite singleton $\{p\}$, and then $\square A=A(p)$  (see \cite{BBL2}). $A(p)$ can supply finite descriptions of objects. In logic,  $p$ enables the substitution of  variables by closed terms: for, by definition, the operator $A(p)$, at time $t$, gives the result of the measurement by the observable $\sigma_z$ of the particle in state $s(t)$ represented by $(\forall x\in D_t)A(x)$, that is the  finite information
$(\forall_\omega x\in D_t)A(x)$: we have $A(p)(t)=(\forall_\omega x\in D_t)A(x)$.

The case $\beta_x=1$ is when the direction $d$ is $x$, orthogonal to the direction chosen for the preparation of the particle and the Hermitian matrix of the observable is $\sigma_X$, the off-diagonal real unitary matrix that switches the eigenvectors of $\sigma_Z$ (and conversely). Reasoning by analogy with the former case, it  can be written as the real linear combination of the two antiprojectors, namely the couple of operators one answering “down" when applied to “up" and the other conversely. Then, overall, we interpret the modal operator $\square_x$ as the abstract antiprojector, and label it $\square_n$. Let us characterize the result of the abstract antiprojector: the negative witness $n$. Then $\square_n A=A(n)$. Since in quantum mechanics two states are distinguishable if and only if they are orthogonal,   and since the results of the projector and of the corresponding antiprojector are orthogonal, $z=n$ means $z\neq p$.  The pair $(p,n)$ is an abstract form of the pairs $(\uparrow_d,\downarrow_d)$, that all together define the constant $\perp$ as seen above. 
Indeed, in the model, one finds the following form of non contradiction:
 $$
\square A, \square_n A\vdash \perp
$$ 
that we could consider as a “modal uncertainty", since it comes out to be the translation of the incompatibility between the observables $\sigma_z$ and $\sigma_x$ (see also the final part of the paper for this),  as proved in \cite{BBL2}:

\begin{theorem}
The non contradiction law, under the form  
 $$
\square A, \square_n A\vdash \perp
$$  
can be interpreted in the quantum model of modal operators.
\end{theorem}
\begin{proof}
    Let us write $\square A$ in its finite form $A(p)$ and rewrite $\square_n A$ as $A(n)$. Thus, modal uncertainty is rewritten $A(p), A(n)\vdash \perp$. On the other side, as seen above, in its infinite form, $\square A$ is $(\forall x \in T)A(x)$, where $T$ is an infinite singleton. Then, let us assume both $A(p)$ and $A(n)$. Then, if $z$ is the generic unique element of $T$, one has both $z=p$ and $z =n$. We can adopt the same letter $z$ since any two unspecified elements chosen in an infinite singleton can be proved to be equal, as implied by the definition of infinite singleton:  for, $(\exists x\in T)x=z$ is equivalent to $(\forall x\in T)x=z$. Then, since $z =n$ means $z\neq p$, one finally has the equivalent writing:
$z=p, z\neq p\vdash \perp$
which is an instance of non contradiction, expressed under the form of the law of identity.
\end{proof}

In conclusion the modality $\square_n$ includes a negation. $\square_n$ answers with data  contradicting  the data from reality: at at any time $t$, it gives the statistics of the state orthogonal to $s(t)$: $A(n)(t)= (\forall_\omega x\in D^\perp) A(x)$, where, if $D_t=\{(\uparrow, \alpha), (\downarrow, \beta)\}$,  $D_t^\perp=\{(\uparrow,\beta), (\downarrow,\alpha)\}$, if $D_t=\{(\uparrow,1)\}$, $D_t^\perp=\{(\downarrow,1)\}$, if $D_t=\{(\downarrow,1)\}$, $D_t^\perp=\{(\uparrow,1)\}$.   
We can make such a negation explicit by defining a negation $\neg$ on modal formulae $\square A$ as follows:
\begin{equation}\label{defneg}
\neg \square A \equiv \square_n A    
\end{equation}
Notice that then, if the infinite interpretation of $\square$ is considered, the negation $\neg$ acts as a “finitizer". Then double negation does not assert.  For, let us apply $\neg$ to $\neg \square A =\square_n A$: this is conceivable since, for any $d$, $\square_d A=\square\square_d A$, is closed, as seen above. Then, $\neg \square_n A=\neg\neg \square A$ is not $\square A$, but only its finite part. For $\square_n A$ is finite and $\neg$ cannot recapture the infinite content of $\square A$.
This accounts for the intuitionistic interpretation of negation.

In the psychoanalytic interpretation,
 the idea of negation as a finitization agrees  with the Freudian idea that negation and contradiction characterize the advent of the Secondary Process, that can take into account the finite elements coming from reality, whereas negation in the Primary Process  is impossible, \cite{Fr00}.
 Moreover, one finds the exact correspondence with the idea of negation as discussed by Freud in his paper {\em Negation} ({\em Die Verneinung}), \cite{Fr25} . The Freudian conception considers negation as the
intellectual counterpart of repression.
In the quantum model, any object consciously represented by the mind
(word presentation) corresponds to an eigenstate of $\sigma_Z$ . Repression means
that the conscious representation is forgotten and substituted by an
unconscious thing-presentation, which includes the opposite of the conscious one by condensation, and hence by
the superposition of the eigenstates in the quantum model. The last is an
eigenstate of $\sigma_X$. Then, in order to find out the object, the observable $\sigma_X$ must be applied.
In the original computational basis, this means the creation of the abstract antiprojector, that is, the negation operator (see \cite{BBL2}). 

In view of this, we could term $\neg$ so defined “neurotic negation" ({\em Verneinung}), given by the contact with reality and its successive repression, in opposition to a “psychotic negation" ({\em Verleugnung}) that, on the contrary, is derived from the failure of the contact with reality, as we see in the next case.

The  two finite modal components given by $\sigma_Z$ and $\sigma_X$, described by the positive witness $p$ and by the negative witness $n$ respectively, enables speaking, i.e. to form {\em word-presentations} of objects previously obtained as infinite {\em thing-presentations}.  For, as seen above, $A(p)(t)=(\forall_\omega x\in D_t)A(x)$ and $A(n)(t)=(\forall_\omega x\in D_t^\perp)A(x)$, that is: substitution by closed terms and substitution by closed terms denoting the opposite - the last described by Freud in his paper {\em Die Verneinung}, \cite{Fr25}.
In Bionian terms, the finite attitude is related to the so-called {\em Convivial} Container-Contained relationship.

\subsection{The irreal case}
The case $\beta_y=1$ is for the other  direction  $y$, orthogonal with respect to the direction chosen for the preparation. The Hermitian matrix is $\sigma_Y$, the off-diagonal unitary matrix with imaginary entries, that can be thought as a linear combination of the two antiprojectors,  with imaginary coefficients. Then it should characterize a non-real abstract antiprojector $\square_y$, given by the equation \ref{defmodfin} initialized by $\sigma_y$. A contrast is created with respect to the two finite cases $\sigma_z,\sigma_x$. In the Bloch Sphere, given a preparation in the $z$ direction, the qubits in the $xz$ plane differ in the statistics they produce but have relative phase $0$ or $\pi$ corresponding to real numbers $\pm 1$, whereas the qubits in other directions have a different relative phase,  that cannot be given by the measurement.  Out of that plane, the “real" value, that is the phase, is not real indeed!

  The non-real abstract antiprojector $\square_y$ can characterize the third abstract element $e$ (for “empty"), such that $\square_y A= A(e)$. We maintain that $e$ is the rejection witness, 
inducing the rejection of the contact with reality, namely the Freudian {\em Verleugnung}, rather than  the negation of the real value (the Freudian {\em Verneinung}), seen above: as we have just pointed out, the {\em Verneinung} case is the case of the witness $n$ for the real abstract antiprojector $\sigma_n$ associated to $\sigma_X$.

The  {\em Verleugnung} operator $\square_y$ and its witness $e$ can be interpreted in Bion's theory of knowledge \cite{BiL}, as the negative element impeding representations, in contrast  with the positive element $p$ and with the whole container $T$ which make them possible. In particular, in Bionian terms, they correspond to the {\em Parasitic} Container-Contained interaction, as if the Container has “squeezed out" its content up to when the empty state is reached.

Then we work on the hypothesis that $\sigma_Y$,  creating a different kind of negation, gives the opposite of the infinite aspect of the modality $\square$, as given in \ref{defmodinf}. Actually, among the four elements of the basis in decomposition \ref{decomp}, we  distinguish two pairs: the first is $\sigma_Z, \sigma_X$, namely the pair of real and finite opposites. The other is $I, \sigma_Y$. 
It is an infinite/irreal pair, that associates “the total knowledge" with “no knowledge", namely it represents a total opposition. Actually it is a particular view of the opposition between $1$ and $\perp$ that we have first described in subsections \ref{psych} and \ref{Bloch}. The last is the opposition between the two abstract Kleinian positions, the Depressive and the Paranoid-Schizoid respectively, while the “total knowledge" corresponds to the knowledge of a given object (a quantum particle in our quantum model) that is given by the infinite interpretation of the  modal operator $\square$. The no-knowledge operator is described as follows.
We  work on the hypothesis that the opposite modalities given by $I, \sigma_Y$ share similar features. Then let us assume that, as in the case of the identity, the parameter $t$ in equation \ref{defmodfin} for $\sigma_Y$, can be imported as an inner variable, finding something like \ref{defmodinf}. One could think that the parameter is “re-swallowed" as a variable since there is no possibility of instantiation with closed terms coming as data from reality. Then the case is non-finite, that is non-non-infinite. It does not mean infinite: the infinite cannot be recovered, once it is lost. It differs with respect to the case \ref{defmodinf} in the fact that the domain $T$ is converted into a dual domain, $N$,  that is the container of failures. Let us label $\cancel{\square}$  the para-infinite operator that corresponds to $\square_y$ and put its equation:
\begin{equation}\label{defimp}
\square \Gamma \vdash \cancel{\square} A \;\mbox{if and only if }\;\Gamma, z \in N\vdash A(z)
\end{equation}
The modality  $\cancel{\square} $ means “impossible", here defined as a primitive notion. Then it hides a negation, coming from the {\em Verleugnung}, the psychotic negation. In logic, it can emerge under the form of duality, hiding the rejection element $e$ and  the set of failures $N$, as we show below.

\section{Shifting from pre-logic to logic}
\subsection{Duality}\label{duality}
According to  basic logic \cite{SBF}, connectives come in dual pairs since each pair corresponds to a unique metalinguistic link that is imported into the object language on the right  or on the left of the turnstyle $\vdash$, putting pairs of symmetric equations:  each equation has its symmetric form, deriving the dual connective. Actually, the orientation of the turnstyle $\vdash$ amounts to the separation of the two elements of the pair.

Quantifiers are all derived from the metalinguistic link {\em forall}. 
When {\em forall} is directly imported as a connective, adopting parameters, the equation are \ref{defforallomega} for $\forall_\omega$ and \ref{defexistsomega} for its dual $\exists_\omega$. As in the case of propositional connectives and constants (see \cite{SBF}), the dual pair $\forall_\omega, \exists_\omega$ emerges simply by switching the left and the right side with respect to the turnstyle.
The case of the dual pair $\forall$ and $\exists$ is  different. When {\em forall} is imported in order to refer to a variable of the object language, one has equation \ref{defforall} for the universal quantifier $\forall$. As for its dual, the existential quantifier $\exists$, according to its intended meaning, the equation is \ref{defexists}, \cite{MS}, while
 the switching  of equation \ref{defforall}, consistent with  \ref{defexists},  is the following:
\begin{equation}
(\exists x\in D)A(x)\vdash \Delta(-z) \;\mbox{if and only if}\; A(z)\vdash z\notin D, \Delta(-z)   
\end{equation}
It requires the exclusion of a membership: $z\notin D$.  
As we noticed, to exclude a membership is not proper of the Unconscious and hence here something different can really arise.
Let us abstract the above equation, as in the case of the universal quantifier:
\begin{equation}\label{defmoddual}
\Diamond A\vdash \Diamond\Delta \;\mbox{if and only if}\; A\vdash \Diamond \Delta
\end{equation}
defining the modal operator $\Diamond$, that is possibility,
and then consider its nature of abstract quantifier, that is:
\begin{equation}
(\exists x\in T)A(x)\vdash \Diamond \Delta \;\mbox{if and only if}\; A(\xi)\vdash \xi\notin T, \Diamond \Delta 
\end{equation}
then converted into:
\begin{equation}
(\exists x\in T)A(x)\vdash \Diamond \Delta \;\mbox{if and only if}\; \xi\in T, A(\xi)\vdash \Diamond \Delta    
\end{equation}
introducing the notion of possibility as an abstract quantifier on the domain $T$. Actually it states “there is a possibility/there is a direction/a datum is found - in the container $T$", determining but hiding at the same time  the set of failures $N$.

Once a negation is explicitly introduced,  the finite negation $\neg$ derived from the negative operator $\square_n$ in \ref{defneg}, impossibility can be derived by negation. Formally, $(\exists x\in T)A(x)\vdash $ is converted into $\vdash \neg (\exists x\in T)A(x)$. The last, according to the meaning of the existential quantifier, is  $\vdash \neg \exists x(x \in T)\& A(x)$, that is $\vdash \forall x \neg (x \in T \& A(x))$. Now, which conjunct should we negate? The most primitive form of negation available is $x\in N$, the pattern for rejection, that means $x\notin T$. Then we get $\vdash (\forall x\in N)A(x)$, namely $\vdash \cancel{\square} A$. So the negative content of $\sigma_Y$ is translated, in logic, into the emergence of the dual connective $\Diamond$, that is introduced on the left of $\vdash$, 
according to its definition \ref{defmoddual}. The need to separate 
$\square$ and $\cancel{\square}$, namely $T$ and $N$, yields the actual separation of the two dual connectives $\square$ and $\Diamond$, and hence, according to their definition as abstract quantifiers, the possibility to conceive sets that are not singletons. For, if $\square$ is not $\Diamond$, $T$ cannot be a singleton any more. Necessity is  one, the world is is full of possibilities (or: no possibility is found). This is the pattern that carries us out of the symmetric world.

\subsection{Looking for integration}
Let us assume that, from the contact with reality, in a certain context denoted by $\Gamma$, elements $\beta_1,\dots ,\beta_n$ are detected at time $t$. Prior to representation, such elements are not present in the mind yet, they are invisible to the mind. Then, from the point of view of the mind, let us  let us describe such a situation by
$$\Gamma\vdash {\cancel\beta_1},\dots ,{\cancel \beta_n}$$
where the ${\cancel\beta_i}$ have a negative value, they are missing.
Let us see how the mind might reach a representation.
One can first assume that, prior to representation, the process of detection is $\sigma_Y$-dominated. Then, in analogy with the position for the overall abstract modal operator ${\cancel \square}$, let us put the equivalence
$$
\Gamma\vdash (\forall i\in N_t){\beta}_i \;\mbox{if and only if }\;
\Gamma\vdash {\cancel\beta_1},\dots ,{\cancel \beta_n}
$$
where  $N_t$ is a null domain and $\forall i\in N_i$ is a fake  quantification, since no variable is present in order to gather the $\beta_i$ yet, except the indices, that are not part of the $\beta$s anyway. So a non-object is created out of the $\beta$s, since $(\forall i\in N_t){\beta}_i $ is like a monster, namely an object of disassembled parts (see \cite{BiL,BiE}).
The goal for the representation is to achieve an assembled object, associated to the form
$$
\Gamma(-z)\vdash (\forall x\in D_t)A(x)
$$
namely: there is an internal variable $z$, forming a non-null domain $D_t$ gathering all the $\beta$s  that  are then linked by the formula $A(z)$ with free variable $z$, for which $\Gamma(-z)$ is a common context.
In order to achieve it, let us continue the analogy with the $\cancel \square$ case, namely let us assume such a case as an abstract pattern for the mind. So let us assume that, in order to avoid the null domains and fake quantifications, $\Gamma\vdash (\forall i\in N_t)\beta_i$ is dualized,  becoming
$$
(\exists x\in D_t)A(x)\vdash \Delta(-z)
$$
assuming that, in the dualization process, an internal variable $z$ is borrowed elsewhere (the variable cannot be internally generated in the mental world but it must be somehow received from the external reality) and introduced, allowing the substitution of the null domain $N_t$ with a “good" non-empty domain $D_t$. Then, in the symmetric mode proper of the Unconscious,  $D_t$ is conceived as an infinite singleton and $\vdash$ has no orientation, so that one can achieve the representation, described as usual by
$$
\Gamma(-z)\vdash (\forall x\in D_t)A(x)
$$
Finally, once the representation is achieved, one can name the $\beta$s originally detected (for the sake of simplicity we assume that the name is $\beta$) and so get the finite representation $\Gamma\vdash (\forall_\omega x\in D_t)A(x)$, namely
$$
\Gamma\vdash A(\beta_1)\&\dots A(\beta_n)
$$
describing all the items. 

\subsection{More on duality}
Notice that, once the whole process of representation is completed, $\square$ and $\Diamond$ are distinct dual operators. Then, at this level,  one has better to consider that the unknown value of a given object is better described by $\Diamond A\equiv (\exists x\in T)A(x)$ rather than by $\square A\equiv (\forall x\in T)A(x)$, as at the stage preceding the distinction, since in origin the domain $T$ contains also the Parasitic part that should be separated. 
Then we can conclude our introduction of  duality by considering the dual of non contradiction  $\square A, \neg \square A\vdash \perp$, considered in subsection \ref{zx}
\begin{theorem}
The excluded middle law, under the form 
$$1\vdash \neg \Diamond A, \Diamond A$$
can be interpreted in the quantum model of modal operators.     \end{theorem}
\begin{proof}
Let us rewrite $1\vdash \neg \Diamond A, \Diamond A$  as $1\vdash {\cancel \square} A, \Diamond A$. In our quantum interpretation, the constant $1$ is the singlet state. As is well known, once one of the particles of the singlet state has been traced out, namely it can be attributed a state as a single particle, the other is unknowable. We represent the state of a traced-out particle by  $\Diamond A$, the unknowability of a particle by ${\cancel \square} A$. Then the sequent $1\vdash {\cancel \square} A, \Diamond A$ 
describes the true fact that, from the singlet state, one of the two alternatives for any of the two follows (and nobody can know which is which, since the two particles of the pair are indistinguishable).
    
\end{proof}

As a final remark we would like to consider the three commutation relations, that quantum mechanics summarizes by the unique equality $[\sigma_i,\sigma_j]=\pm\dfrac{i}{2}\sigma_k$, where $i,j,k$ is any permutation of $x,y,z$. When the $xz$ plane is characterized in the Bloch sphere (the direction $z$ of the preparation and one direction $x$ in the orthogonal plane), their translation into the pre-logical framework is different and summarized as follows:
\begin{itemize}
    \item the incompatibility between $\sigma_z$ and $\sigma_x$ corresponds to non contradiction, see subsection \ref{zx};
    \item the incompatibility between $\sigma_z$ and $\sigma_y$ corresponds primarily to the opposition between $\square A$ and its contrary ${\cancel{\square}}A$; then,  when duality is adopted, to the weaker form $1\vdash \neg \Diamond A, \Diamond A$, the excluded middle law, as just seen;
    \item the incompatibility between $\sigma_x$ and $\sigma_y$, in rational terms, is avoided reducing it to the two above, by identifying neurotic and psychotic negation, namely by putting a definition of negation after $\sigma_x$, as seen in \ref{zx}, by hiding $\sigma_y$ behind duality, as seen above, and finally  making them converge into the unique negation of classical logic.
\end{itemize}
Indeed, in the construction of our Knowledge, in Bionian terms, we establish the Convivial Link and we  prefer it  with respect to the Parasitic Link, hence we characterize the real $zx$ plane and what follows.

 \subsection{Concluding remarks}
The pre-logical elements so introduced can justify the emergence of finite elements, the separation between assertion and negation and the possibility for a mathematical infinite. The mathematical infinite is nested in the stronger symmetric infinite,  as discussed by Matte Blanco in \cite{MB75}. Analogously, all components of our thinking must be derived, a non-trivial task, from the pure realm of infinite singletons, where such components cannot apply since Symmetry and strong coherence dominate.  We have seen how pre-logical elements can be nested in the realm of infinite singletons, thanks to the introduction of modal operators that allow to shift from the mode of the Unconscious to rational thinking, as we have proved. We would like to recall that, in the evolution of the Freudian theory, the moving from the First to the Second Topic, as expressed in The Ego and the Id \cite{Fr23}, introduces a normative instance moderating the encounter of the psychic dimension with the external reality. Hence, the pre-logical elements must be conceived as the seeds of logical elements, not a logical system itself.
Once the logic of standard infinite sets, standard finite sets, standard dual connectives, negation and contradiction, is established,  the pre-logical possibilities offered by the decomposition \ref{decomp} are partially lost. For, finding a way to recover the full complete information from a finite incomplete one, abstracting  definition \ref{defforallomega} to obtain \ref{defforall}, requires, according to the present model, to  re-integrate the pre-logical component. Indeed, as seen, in order to obtain a variable, one needs the infinite component, on one side,  and the hidden psychotic negation ({\em Verleugnung}), corresponding to the irreal observable $\sigma_y$ representing the failure of the representation, on the other. As seen, such non-finite components   act in the shift from the finite omega-quantifier to the abstract quantifier:  in quantum mechanics, the quantifier represents the pure state. The point is  the relative phase of the qubit, that makes the Bloch sphere three-dimensional, so that the direction is not derivable from a finite information, namely from the statistics given by the measurement process. 
 
All this does not exclude further possibilities in the logical analysis. From a theoretical point of view in logic, there are many intriguing questions. In terms of the sequent calculi forming  the “cube of logics" given in Basic Logic \cite{SBF}, we can formulate some issues, such as, for example: 
\begin{itemize}
    \item the role of intuitionistic logic as a better representative of the “logic of the infinite" with respect to classical logic;
    \item the role of dual intuitionistic logic as a shadow-logic created by the necessity to hide Bion's Parasitic Link;
    \item the role of the basic sequent calculus $\bf B$, which is so weak that its two negations, namely a weak kind of intuitionistic negation, on one side, and  negation given by exclusion, on the other, cannot identify in it, that makes it an interesting platform; 
    \item the study of the structural rules of sequent calculus as rules derived by the need to integrate, see the above subsection;
    \item the role of multiplicative negation as an alternative way to deal with  {\em bags of symmetry}  (Matte Blanco's terminology)
\end{itemize}

In general, one could now even question if the dual language, widespreadly adopted in formalizing logic, probably since it shows useful in hiding irrational elements, could be misleading, for, at the same time, a lot of very useful pre-logical elements are hidden by it. The apparent example is that of intuitionistic logic, which is not dual by nature, however its standard formalization is damned by a dual language. Analogously, the standard logical formalization is damned by first order language, in which separating closed terms from variables is possible only at the metalevel, see \cite{BBL1}. Up to now, by the present proposal, we can point out what, where and why is hidden behind, that is indispensable, we think, to further developments in the right direction. In our opinion, the development of any effective pre-logical basis should include  an explicit account of infinite singletons inside the language, that enables to overcome both dual language and first order terms at the same time. As we have seen, there is an infinite reading of the modality $\square$ that can overcome both. However, quickly summarized in quantum terms, it cannot include the entanglement, at least in the present reading: for, in defining the modality, we have assumed to characterize a particle. Further developments, modal or not, require to include it.  Among the different existing logical frameworks, the varieties of approaches derived from Linear Logic \cite{Gi} should  allow better insights if considered in light of our issues. A different logical framework useful to the analysis is that provided by the Square of Oppositions, also in its modal versions, see the volume \cite{BB}.

Further insights should derive by a further development of the formal approaches to the psychoanalytic models, see \cite{BBL1}. For example, considering the role of the Kleinian {\em Projective Identification}.  This would carry new possibilities in applications, assuming that {\em the unconscious roots of AI lie in a form of projective identification, i.e., an emotional and imaginative exchange between humans and machines.}, as claimed in \cite{Po}.
Any kind of advancing in the logical comprehension would support, as for the field of A.I., applications such as that already proposed in \cite{HH}: 
{\em we attempt a quantum‐computational construction of robot affect, which theoretically should be able to account for indefinite and ambiguous states as well as parallelism.}

\end{document}